\documentclass[11pt]{article}
\usepackage{fullpage}
\usepackage{authblk}
\usepackage{amsfonts}
\usepackage{amssymb}
\usepackage{amsmath}
\usepackage{amsthm}
\usepackage{latexsym}
\usepackage{graphicx}
\usepackage{float}
\usepackage[breaklinks]{hyperref}

\newcommand{\ADVpm}{\mathrm{ADV}^\pm}

\newcommand{\G}{\mathcal{G}}
\newcommand{\V}{\mathcal{V}}
\newcommand{\E}{\mathcal{E}}
\newcommand{\N}{\mathbb{N}}

\long\def\rem#1{}
\def\B{\{0,1\}}

\newcommand{\braket}[2]{\langle#1|#2\rangle}
\newcommand{\norm}[1]{\| #1 \|}

\newcommand{\D}{\mathcal{D}}
\newcommand{\LG}{\mathcal{LG}}

\newtheorem{definition}{Definition}
\newtheorem{theorem}{Theorem}
\newtheorem{lemma}[theorem]{Lemma}

\newcommand{\thmref}[1]{\hyperref[#1]{{Theorem~\ref*{#1}}}}
\newcommand{\lemref}[1]{\hyperref[#1]{{Lemma~\ref*{#1}}}}
\newcommand{\corref}[1]{\hyperref[#1]{{Corollary~\ref*{#1}}}}
\newcommand{\eqnref}[1]{\hyperref[#1]{{Equation~(\ref*{#1})}}}
\newcommand{\claimref}[1]{\hyperref[#1]{{Claim~\ref*{#1}}}}
\newcommand{\remarkref}[1]{\hyperref[#1]{{Remark~\ref*{#1}}}}
\newcommand{\propref}[1]{\hyperref[#1]{{Proposition~\ref*{#1}}}}
\newcommand{\factref}[1]{\hyperref[#1]{{Fact~\ref*{#1}}}}
\newcommand{\defref}[1]{\hyperref[#1]{{Definition~\ref*{#1}}}}
\newcommand{\exampleref}[1]{\hyperref[#1]{{Example~\ref*{#1}}}}
\newcommand{\hypref}[1]{\hyperref[#1]{{Hypothesis~\ref*{#1}}}}
\newcommand{\secref}[1]{\hyperref[#1]{{Section~\ref*{#1}}}}
\newcommand{\chapref}[1]{\hyperref[#1]{{Chapter~\ref*{#1}}}}
\newcommand{\apref}[1]{\hyperref[#1]{{Appendix~\ref*{#1}}}}

\newcommand{\ignore}[1]{}
\begin{document}
\title{Learning graph based quantum query algorithms\\ for finding constant-size subgraphs%
\thanks{Partially supported by the French ANR Defis project
ANR-08-EMER-012 (QRAC)
and
the European Commission IST STREP project
25596 (QCS). Research at the Centre
for Quantum Technologies is funded by the Singapore Ministry of Education 
and the National Research Foundation.}
}

\author[1]{Troy Lee\thanks{\texttt{troyjlee@gmail.com}}}
\author[2]{Fr\'ed\'eric Magniez\thanks{\texttt{frederic.magniez@univ-paris-diderot.fr}}}
\author[1,2]{Miklos Santha\thanks{\texttt{miklos.santha@liafa.univ-paris-diderot.fr}}}
\affil[1]{Centre for Quantum Technologies, National University
    of Singapore, Singapore 117543
    }
\affil[2]
{CNRS, LIAFA, Univ Paris Diderot, Sorbonne Paris Cit\'e, F-75205 Paris, France}
\date{}
\maketitle

\begin{abstract}
Let $H$ be a fixed $k$-vertex graph with $m$ edges and minimum degree $d >0$.  
We use the learning graph framework of Belovs to show that the bounded-error 
quantum query complexity of determining if an $n$-vertex graph contains $H$ as 
a subgraph is $O(n^{2-2/k-t})$, where 
\[
t = \max\left\{\frac{k^2- 2(m+1)}{k(k+1)(m+1)}, \ \frac{2k - d - 3}{k(d+1)(m-d+2)} \right\} > 0 \enspace .
\]
The previous best algorithm of 
Magniez et al.\ had complexity $\widetilde O(n^{2-2/k})$.
\end{abstract}




\section{Introduction}
{\bf Quantum query complexity.}
Quantum query complexity has been a very successful model for studying the power of 
quantum computation.  Important quantum algorithms, in particular the search algorithm 
of Grover~\cite{gro96} and the period finding subroutine of Shor's factoring algorithm~\cite{shor97}, 
can be formulated in this model, yet it is still simple enough that one can often prove 
tight lower bounds.  This model is the quantum analog of deterministic and randomized 
decision tree complexities; the resource measured is the number of queries to the input 
and all other operations are for free.

For promise problems the quantum query complexity can be exponentially smaller
than the classical complexity,
the Hidden Subgroup Problem~\cite{sim97,ehk99} being the most striking example.
The situation is dramatically different for total functions, as
Beals et al.~\cite{bbcmw01} showed that
in this case the deterministic and the quantum
query complexities are polynomially
related. 

One rich source of concrete problems are functions related to properties of graphs.
Graph problems were first studied in the quantum query model by Buhrman et al.~\cite{BCWZ99} 
and later by Buhrman et al.~\cite{bdhhmsw05}, who looked at
Triangle Finding together with Element Distinctness. This was followed by the exhaustive work of 
D\"urr et al.~\cite{dhhm06} 
who investigated many standard graph problems
including Connectivity, Strong Connectivity, Minimum Spanning Tree, and Single Source Shortest Paths.
All these approaches were based on clever
uses of Grover's search algorithm. The groundbreaking work of Ambainis~\cite{amb07}
using quantum walks for Element Distinctness initiated the study  of
quantum walk based search algorithms. Magniez et al.~\cite{MSS07} used this technique to design  
quantum query algorithms for finding constant size subgraphs, and recently 
Childs and Kothari found a novel application of this framework to decide minor-closed graph 
properties~\cite{ck11}.
The results of~\cite{MSS07} imply that a $k$-vertex subgraph can be found with 
$\widetilde{O}(n^{2 - 2/k})$ queries, and moreover
Triangle Finding is solvable with $\widetilde{O}(n^{1.3})$ queries.  Later,
quantum phase estimation techniques~\cite{mnrs11} were also applied to these problems,
and in particular the quantum query complexity of Triangle Finding was improved to 
${O}(n^{1.3})$. The best lower bound known for finding any constant sized subgraph is the trivial 
$\Omega(n)$.

{\bf The general adversary bound and learning graphs.}
Recently, there have been exciting developments leading 
to a characterization of quantum query complexity in terms of a (relatively) simple 
semidefinite program, the general adversary bound~\cite{Reichardt10advtight, lmrss11}.  
Now to design quantum algorithms it suffices to exhibit a solution to this semidefinite 
program.
This plan turns out to be quite difficult as the minimization form of the general adversary 
bound (the easiest form to upper bound) has exponentially many constraints.  Even for 
simple functions it is difficult to directly come up with a feasible solution, much less worry about 
finding a solution with good objective value.  

Belovs~\cite{spanCert} recently introduced the model of learning graphs, which can be viewed as 
the minimization form of the general adversary bound with additional structure imposed on the 
form of the solution.  This additional structure makes learning graphs much easier to reason 
about.  In particular, it ensures that the feasibility 
constraints are {\em automatically} satisfied, allowing one to focus on coming up with a solution 
having a good objective value.
Learning graphs are a very promising model and have already been used to improve the complexity of 
Triangle Finding to $O(n^{35/27})$~\cite{spanCert} and to give an $o(n^{3/4})$ algorithm for $k$-Element Distinctness~\cite{Bel12}, improving the previous bound of $O(n^{k/(k+1)})$ \cite{amb07}.  

{\bf Our contribution.}
We give two learning graph based algorithms for the problem of determining if a graph $G$ contains 
a fixed $k$-vertex subgraph $H$.  Throughout the paper we will assume that $k > 2$, as the problem of determining 
if $G$ contains an edge is equivalent to search.  
We denote by $m$ the number of edges in $H$.  
The first algorithm we give has complexity 
$O(n^{2 - 2/k -t})$ where $t = (k^2- 2(m+1))/(k(k+1)(m+1)) > 0$.  The second algorithm depends on the minimum 
degree of a vertex in $H$.  Say that the smallest degree of a vertex in $H$ is $d >0$.  This is without loss of generality 
as isolated vertices of $H$ can be removed and the theorem applied to the resulting graph $H'$.  The second algorithm 
has complexity $O(n^{2 - 2/k -t})$ where $t = (2k - d - 3)/(k(d+1)(m+2)) > 0$.
Both algorithms thus improve on 
the previous best general subgraph finding algorithm of~\cite{MSS07}, which has complexity 
$\widetilde{O}(n^{2 - 2/k})$.  The first algorithm performs better, for example, on dense regular graphs $H$, while the 
second algorithm performs better on the important case where $H$ is a triangle, having complexity $O(n^{35/27})$,
equal to that of the algorithm of Belovs \cite{spanCert}. 

To explain these algorithms, we first give a high level description of 
the learning graph algorithm in~\cite{spanCert} for Triangle Finding, and its relation to the quantum walk algorithm 
given in~\cite{MSS07}.  The learning graph algorithm in~\cite{spanCert} for Triangle Finding is roughly a 
translation of the quantum walk algorithm on the Johnson graph of~\cite{MSS07} into the
learning graph framework, with one additional twist.  This is to maintain a database not of all edges present in $G$ 
amongst a subset of $r$-vertices but rather a random sample of these edges.  We will refer to this as sparsifying 
the database.  While in the quantum walk world this idea does not help, in the context of learning graphs it leads to a 
better algorithm.  

The quantum walk of \cite{MSS07} works by looking for a subgraph $H' = H \setminus \{v\}$, 
where $v$ is a vertex of minimal degree 
in $H$, and then (using the algorithm for element distinctness) finding the vertex $v$ and the edges linking it to 
$H'$ to form $H$.  Our second learning graph algorithm translates this procedure into the learning graph framework, 
and again applies the trick of sparsifying the database.  Our first algorithm is simpler and translates the quantum 
walk searching for $H$ directly to the learning graph framework, again maintaining a sparsified database.  

The way we apply sparsification 
differs from how it is used in~\cite{spanCert}.  
There every edge slot is taken independently with some fixed probability, while in our case the 
sparse random graphs are chosen uniformly
from a set of structured multipartite graphs whose edge pattern reflects that of the given 
subgraph. The probability space 
evolves during the algorithm, but at every stage the multipartite graphs have a very regular 
degree structure.  
This uniformity of the probability space renders the structure of the learning graph 
very transparent.    

{\bf Related contribution.}
Independently of our work, Zhu~\cite{zhu11} also obtained
\thmref{t:main2}. 
His algorithm is also based on learning graphs, but differs from ours in working with randomly 
sparsified cliques as in the algorithm of Belovs~\cite{spanCert} for Triangle Finding, rather than 
graphs with specified degrees as we do.

\section{Preliminaries}\label{sec:lg}
We denote by $[N]$ the set $\{1,2, \ldots , N\}$.
The \emph{quantum query complexity} of a function $f$, denoted $Q(f)$, is the number of 
input queries needed to evaluate $f$ with error at most $1/3$.  We refer 
the reader to the survey~\cite{hs05} for precise definitions and background.

For a boolean function $f : \D \rightarrow \B$ with $\D \subseteq \{0,1\}^N$, the general adversary 
bound~\cite{HoyerLeeSpalek07negativeadv}, denoted $\ADVpm(f)$, can be defined as follows 
(this formulation was first given in \cite{ReichardtSpan}):
\begin{equation} \label{e:neg_adv_dual}
\begin{aligned}
\ADVpm(f)= \ & \underset{u_{x,i}}{\text{minimize}}
& & \max_{x \in \D} \ \sum_{i \in [N]} \norm{u_{x,i}}^2 \\
& \text{subject to}
& & \sum_{\substack{i \in [N] \\ x_i \ne y_i}} \braket{u_{x,i}}{u_{y,i}} =1 \text{ for all } f(x) \ne f(y)
 \enspace .
\end{aligned}
\end{equation}

As the general adversary bound characterizes quantum query complexity 
\cite{Reichardt10advtight}, 
quantum algorithms can be developed (simply!) by devising solutions to this semidefinite program.  
This turns out not to be so simple, however, as even coming up with 
{\em feasible} solutions to \eqnref{e:neg_adv_dual} is not easy because of the large number of 
strict constraints.  

Learning graphs are a model of computation introduced by Belovs~\cite{spanCert} that give rise to 
solutions of \eqnref{e:neg_adv_dual} and therefore quantum query algorithms.  The model of 
learning graphs is very useful as it ensures that the constraints are satisfied automatically, allowing one to focus 
on coming up with a solution having a good objective value. 

\begin{definition}[Learning graph]
A learning graph $\G$ is a 5-tuple $(\V, \E, w, \ell, \{p_y: y \in Y\})$
where $(\V,\E)$ is a  
rooted, weighted and directed acyclic graph,  the weight function 
$w : \E \rightarrow \mathbb{R}$ maps learning graph 
edges to positive real numbers, the length function $\ell:\E \rightarrow \N$ assigns each edge a 
natural number, and $p_y : \E \rightarrow \mathbb{R}$   
is a unit flow whose source is the root, for every $y \in Y$.
\end{definition}

\begin{definition}[Learning graph for a function]
Let $f : \B^N \rightarrow \B$ be a function.  A learning graph $\G$ for $f$ is a 5-tuple $(\V, \E,S, w, \{p_y : y \in f^{-1}(1)\})$, where 
$S: \V \rightarrow 2^N$ maps $v \in \V$ to a set $S(v)\subseteq [N]$ of variable indices,
and $(\V, \E, w, \ell, \{p_y: y \in f^{-1}(1)\})$ is a learning graph 
for the length function $\ell$ defined as $\ell((u,v)=|S(v) \setminus S(u)|$ for each edge $(u,v)$.
For the root $r \in \V$ we have $S(r) = \emptyset$, 
and every learning graph edge $e=(u,v)$ satisfies $S(u) \subseteq S(v)$.
For each input $y \in f^{-1}(1)$, the set $S(v)$ contains a $1$-certificate for $y$ on $f$, for every sink 
$v \in \V$ of $p_y$. 
\end{definition}

Note that it can be the case for an edge $(u,v)$ that $S(u)=S(v)$ and the length of the edge is zero.
In Belovs \cite{spanCert} what we define here is called a reduced learning graph, and a learning graph is 
restricted to have all edges of length one.  

\begin{definition}[Flow preserving edge sets]
A set of edges $E \subseteq \E$ is {\em flow preserving}, if in the subgraph $G=(V,E)$ 
induced by $E$, for every vertex $v \in V$ which is not a source or a sink in $G$, 
$\sum_{u \in V} p_y((u,v)) = \sum_{w \in V} p_y((v,w))$, for every $y$.
For a flow preserving set of edges $E$ we let $p_y(E)$ denote {\em the value of the flow $p_y$ over $E$}, that is 
$p_y(E) = \sum_{s : \text{source in } G} \sum_{v \in V} p_y((s,v))$. 
\end{definition}

Observe that 
$p_y/p_y(E)$ is a unit flow over $E$ whenever $p_y(E) \ne 0$, and that $p_y(\E) = 1$ for every $y$.
The complexity of a learning graph is defined as follows.
\begin{definition}[Learning graph complexity] \label{e:lg_cost}
Let $\G$ be a learning graph, and let $E \subseteq \E$ a set of flow preserving learning graph edges.
The negative complexity of $E$ is $C_0(E)=\sum_{e \in E} \ell(e) w(e)$. The positive complexity
of $E$ under the flow $p_y$ is 
\begin{equation*}
C_{1,y}(E)= 
\sum_{e \in E}  \frac{\ell(e)}{w(e)}\Big(\frac{p_y(e)}{ p_y(E)}\Big)^2, \text{ if $p_y(E) > 0$,\quad and $0$ otherwise.}
\end{equation*}
The positive complexity of $E$ is $C_1(E) =\max_{y \in Y} C_{1,y}(E).$
The complexity of $E$ is $C(E)=\sqrt{C_0(E) C_1(E)}$, and 
the learning graph complexity of $\G$ is $C(\G)= C(\E)$.  
The learning graph 
complexity of a function $f$, denoted $\LG(f)$, is the minimum learning graph complexity of 
a learning graph for $f$.
\end{definition}

The usefulness of learning graphs for quantum query complexity is given by the following theorem.
\begin{theorem}[Belovs] \label{t:lg_alg} 
$Q(f)=O(\LG(f))$.  
\end{theorem}

We study functions $f: \B^{n \choose 2} \rightarrow \B$ whose input is an undirected $n$-vertex graph.
We will refer to the vertices and edges of the learning graph 
as $L$-vertices and $L$-edges so as not to cause confusion with the vertices/edges of the input graph.
Furthermore, we will only consider learning graphs where every $L$-vertex is labeled by a $k$-partite 
undirected graph on $[n]$, where $k$ is some fixed positive integer. Different $L$-vertices will have different labels, 
and we will identify an $L$-vertex with its label.

\section{Analysis of learning graphs}\label{sec:lemma}
We first review some tools developed by Belovs to analyze the complexity of learning graphs and then develop some 
new ones useful for the learning graphs we construct.  We fix for this section a learning graph 
$\G = (\V, \E, w, \ell, \{p_y\})$. 
By  {\em level $d$} of $\G$ we refer to the set of vertices at distance 
$d$ from the root.  A {\em stage} is the set of edges of $\G$ between level $i$ and level $j$, for some $i<j$.  
For a subset $V \subseteq \V$ of the $L$-vertices let $V^+ = \{ (v,w) \in \E : v \in V\}$ and similarly let 
$V^- = \{ (u,v) \in\E : v \in V\}$.  For a vertex $v$ we will write $v^+$ instead of ${\{v\}}^+$, and similarly for $v^-$ 
instead of ${\{v\}}^-$.  Let $E$ be a stage of $\G$ and let $V$ be some subset of  the $L$-vertices at the beginning 
of the stage.  
We set $E_V^{\rightarrow} = \{ (v,w) \in E : \text{$v$ is $u$ or a descendent of $u$ for some $u \in V$}\}$.
For a vertex $v$ we will write $E_v^{\rightarrow}$ instead of $E_{\{v\}}^{\rightarrow}$.

Given a learning graph $\G$, 
the easiest way to obtain another learning graph 
is to modify the
weight function of $\G$. We will often use this reweighting scheme to obtain learning graphs with better complexity 
or complexity that is more convenient to analyze.
When $\G$ is understood from the context, and when $w'$ is the new weight function, for any subset $E \subseteq \E$
of the $L$-edges, we denote the complexity of $E$ with respect to $w'$ by $C^{w'}(E)$.

An illustration of the reweighting method is the following lemma of Belovs which states that 
we can upper bound the complexity of a learning graph by partitioning it into a 
constant number of stages and summing the complexities of the stages.  

\begin{lemma}[Belovs]\label{l:stages}
If $\E$ can be partitioned into a constant number $k$ of stages $E_1, \ldots, E_k$, 
then there exists a weight function $w'$ such that
$C^{w'}(\G)=O(C(E_1) + \ldots + C(E_k))$.
\end{lemma}  

Now we will focus on evaluating the complexity of a stage.  Belovs has given a general theorem to simplify 
the calculation of the complexity of a stage for flows with a high degree of symmetry 
(Theorem 6 in \cite{spanCert}).  Our flows will possess this symmetry but rather than apply Belovs' theorem, 
we develop one from scratch that takes further advantage of the regular structure of our learning graphs.  


\begin{definition}[Consistent flows] 
Let $E$ be a stage of $\G$ 
and let $V_1, \ldots, V_s$ be a partition of
the $L$-vertices at the beginning of the stage.
We say that $\{p_y\}$ is consistent with $E_{V_1}^{\rightarrow} , \ldots, E_{V_s}^{\rightarrow} $ if  
$p_y(E_{V_i}^{\rightarrow} )$ is independent of $y$ for each $i$.

\end{definition}

\begin{lemma}
\label{lem:conv_comb}
Let $E$ be a stage of $\G$ 
and let $V_1, \ldots, V_s$ be a partition of
the $L$-vertices at the beginning of the stage.  Set $E_i = E_{V_i}^{\rightarrow} $, and 
suppose that $\{p_y\}$ is consistent with $E_1, \ldots, E_s$. 
Then there is a new weight function $w'$ for $\G$ such that $$C^{w'}(E) \leq \max_i C(E_i).$$
\end{lemma}

\begin{proof}
Since by hypothesis $p_y(E_i)$ is independent from $y$, denote it by $\alpha_i$.  We assume that $\alpha_i >0$ 
for each $i$; if $\alpha_i =0$ then $p_y((u,v))=0$ for every $y$ and $(u,v) \in E_i$, and these edges can be 
deleted from the graph without affecting anything.  
For $e \in E_i$, we define the new weight $w'(e) = \alpha_i C_1(E_i)w(e)$.  Let us analyze the complexity of $E$ 
under this weighting.  

To evaluate the positive complexity observe that $p_y(E) =1$ for every $y$, since $E$ is a stage,
and thus $\sum_i \alpha_i = 1$. Therefore
\[
C_1^{w'}(E) =
\max_y \sum_i \sum_{e \in E_i}  \frac{\ell(e) p_y(e)^2}{w'(e)} \leq
\sum_i \frac{\alpha_i}{C_{1}(E_i)} \max_y \sum_{e \in E_i}  \frac{\ell(e) p_y(e)^2}{ w(e) \alpha_i^2} 
= \sum_i \alpha_i = 1.
\]

The negative complexity can be bounded by
\[
C_0(E) =
\sum_i \sum_{e \in E_i} \ell(e) w'(e) = 
\sum_i \alpha_i C_1(E_i) \sum_{e \in E_i}\ell(e)  w(e) = 
\sum_i \alpha_i C_1(E_i) C_0(E_i) \le 
\max_i C(E_i)^2.
\]
\end{proof}

At a high level, we will analyze the complexity of a stage $E$ as follows.  First, we partition the set of vertices $V$ 
into equivalence classes $[u]= \{\sigma(u): \sigma \in S_n\}$
for some appropriate action of $S_n$ that we will define later,
and use symmetry to argue that the flow is consistent with
$\{ E_{[u]}^{\rightarrow} \}$.  Thus by \lemref{lem:conv_comb}, it is enough to focus 
on the maximum complexity of $E_{[u]}^{\rightarrow}$. 
Within $E_{[u]}^{\rightarrow}$, our flows will be of a 
particularly simple form.  In particular, incoming flow will be uniformly distributed over a subset of $[u]$ of 
fixed size independent of $y$.  The next two lemmas evaluate the complexity of $E_{[u]}^{\rightarrow}$ in 
this situation.

\begin{lemma}\label{e:simple_cost}
Let $E$ be a stage of $\G$ and let $V$ be some subset of  the $L$-vertices at the beginning of the stage.  
For each $y$ let $W_y \subseteq V$ be the set of vertices in $V$ which receive positive flow under $p_y$.
Suppose that for every $y$ the following is true:
\begin{enumerate}
\item
$E_u^{\rightarrow} \cap E_v^{\rightarrow} = \emptyset$ for $u \neq v \in V$,
\item
$|W_y| $ is independent of $y$, 
\item
for all $v \in W_y$ we have $p_y(E_v^{\rightarrow}) = p_y(E_V^{\rightarrow})/|W_y|$.
\end{enumerate}
Then
$$
C(E_V^{\rightarrow}) \leq \sqrt{\max_{v \in V}C_0(E_v^{\rightarrow}) \max_{v \in V}C_1(E_v^{\rightarrow}) \frac{|V|}{|W_y|}} \enspace .
$$
\end{lemma}

\begin{proof}
The negative complexity can easily be upper bounded by
$$
C_0(E_V^{\rightarrow}) = \sum_{v \in V}C_0(E_v^{\rightarrow}) \leq |V| \max_{v \in V}C_0(E_v^{\rightarrow}).
$$
For the positive complexity we have 
\begin{eqnarray*}
C_1(E_V^{\rightarrow}) &=&
\max_y \sum_{v \in W_y} \sum_{e \in E_v^{\rightarrow}}  \frac{\ell(e) p_y(e)^2}{w(e)p_y(E_V^{\rightarrow})^2} \\
&\leq&
\frac{1}{|W_y|^2} \sum_{v \in W_y}  \max_y  \sum_{e \in E_v^{\rightarrow}}   \frac{\ell(e) p_y(e)^2 \omega^2}{w(e)p_y(E_V^{\rightarrow})^2} \\
&\leq&
\frac{\max_{v \in V}C_1(E_v^{\rightarrow})}{|W_y|}.
\end{eqnarray*}
\end{proof}

Observe that when $E$ is a stage between two consecutive levels, that is between level $i$ and $i+1$ for some $i$, 
and $V$ is a subset of the vertices at the
beginning of the stage, then $E_V^{\rightarrow} = V^+$. 
We will use \lemref{lem:conv_comb} in conjunction with \lemref{e:simple_cost} first in this context.
\begin{lemma}
\label{e:simple_costg}
Let $E$ be a stage of $\G$ between two consecutive levels.  Let $V$ be the set of $L$-vertices at the 
beginning of the stage and suppose that each $v \in V$ has outdegree $d$ and all $L$-edges $e$ of the stage 
satisfy $w(e)=1$ and $\ell(e) \le \ell$.  
Let $V_1, \ldots, V_s$ be a partition of~$V$, 
and for all $y$ and $i$,
let $W_{y,i} \subseteq V_i$ be the set of vertices in $V_i$
which receive positive flow under $p_y$.
Suppose that 
\begin{enumerate}
\item the flows $\{p_y\}$ are consistent with $\{{V_i}^+\}$,
\item  
$|W_{y,i}| $ is independent from $y$ for every $i$, and for all $v \in W_{y,i}$ we have $p_y(v^+) = p_y({V_i}^+)/|W_{y,i}|$,
\item there is a $g$ such that for each vertex $v \in W_{y,i}$ the flow is directed uniformly to $g$ of the $d$ many 
neighbors.  
\end{enumerate}
Then there is a new weight function $w'$ such that
\begin{equation}
\label{e:cost_formula}
C^{w'}(E) \leq   \max_i  \ell \sqrt{\frac{d}{g} \frac{ |V_i|}{|W_{y,i}|}} \enspace .
\end{equation}
\end{lemma}

\begin{proof}
By hypothesis~(1) we are in the 
realm of \lemref{lem:conv_comb} and therefore $C^{w'}(E) \leq \max_i C({V_i}^+)$.
To evaluate $C({V_i}^+)$, we can apply \lemref{e:simple_cost} according to hypothesis~(2).
The statement of the lemma then follows, since for every $v \in V$ we have 
$C_0(v^+) = \ell d$, and $C_1(v^+) = \ell/g$ by hypothesis~(3).
\end{proof}

This lemma will be the main tool we use to analyze the complexity of stages.  Note that the complexity in 
\eqnref{e:cost_formula} can
be decomposed into three parts: the length $\ell$, the degree ratio $d/g$, and 
the {\em maximum vertex ratio} $\max_i |V_i|/|W_{y,i}|$.  This terminology will be very helpful to 
evaluate the complexity of stages.  

We will use symmetry to decompose our flows as a convex combinations of uniform flows over disjoint sets of 
edges.  Recall that each $L$-vertex $u$ is labeled by a $k$-partite graph on $[n]$, say with color classes 
$A_1, \ldots, A_k$, and that we identify an $L$-vertex with its label.  For $\sigma \in S_n$ we define the action of 
$\sigma$ on $u$ as $\sigma(u) = v$, where $v$ is a $k$-partite graph with color classes 
$\sigma(A_1), \ldots, \sigma(A_k)$ and edges $\{\sigma(i), \sigma(j)\}$ for every edge $\{i,j\}$ in $u$.  

Define an equivalence class $[u]$ of $L$-vertices by $[u] = \{ \sigma(u): \sigma \in S_n\}$.  We say that $S_n$ {\em acts 
transitively} on flows $\{p_y\}$ if for every $y, y'$ there is a $\tau \in S_n$ such that $p_y((u,v))=p_{y'}((\tau(u), \tau(v))$ 
for all $L$-edges $(u,v)$. 

As shown in the next lemma, if $S_n$ acts transitively on a set of flows $\{p_y\}$ then they are consistent with 
${[v]}^+$, where $v$ is a vertex at the beginning of a stage between consecutive levels. 
This will set us up to satisfy hypothesis~(1) of \lemref{e:simple_costg}.
\begin{lemma}
\label{lem:trans_cons}
Consider a learning graph $\G$ and a set of flows $\{p_y\}$ such that $S_n$ acts transitively on $\{p_y\}$.  
Let $V$ be the set of $L$-vertices of $\G$ at some given level.
Then $\{p_y\}$ is consistent with $\{[u]^+ : u \in V\}$, and, similarly, $\{p_y\}$ is consistent with $\{[u]^- : u \in V\}$.
\end{lemma}

\begin{proof}
Let $p_y, p_{y'}$ be two flows and $\tau \in S_n$ such that $p_y((u,v))=p_{y'}((\tau(u),\tau(v))$ for all $L$-edges $(u,v)$.
Then 
\begin{align*}
p_y([u]^+) &= \sum_{v \in [u]} \sum_{w: (v,w) \in \E} p_y((v,w)) \\
&= \sum_{v \in [u]} \sum_{w: (v,w) \in \E} p_{y'}( (\tau(v),\tau(w)) ) \\
&=\sum_{\tau^{-1}(v) \in [u]} \ \sum_{\tau^{-1}(w): (\tau^{-1}(v), \tau^{-1}(w)) \in \E} p_{y'}((v,w)    ) \\
&=\sum_{v \in [u]} \sum_{w: (v,w) \in \E} p_{y'}((v,w))
=p_{y'}([u]^+).
\end{align*}
The statement $p_y([u]^-)=p_{y'}([u]^-)$ follows exactly in the same way.
\end{proof}
The next lemma gives a sufficient condition for hypothesis~(2) of \lemref{e:simple_costg} to be satisfied.  
The partition of vertices in \lemref{e:simple_costg} will be taken according to the equivalence classes $[u]$.  
Note that unlike the previous lemmas in this section that only consider a stage of a learning graph, this lemma 
speaks about the learning graph in its entirety.  

\begin{lemma}\label{e:symmetry}
Consider a learning graph and a set of flows $\{p_y\}$ such that $S_n$ acts transitively on $\{p_y\}$.  
Suppose that for every $L$-vertex $u$ and flow $p_y$ such that $p_y(u^{-}) > 0$,
\begin{enumerate}
  \item 
   the flow from $u$ is uniformly directed to $g^+([u])$ many neighbors,
  \label{esh2}
  \item 
  for every $L$-vertex $w$, 
  the number of incoming edges with from $[w]$ to $u$ is $g^-([w], [u])$.  
  \label{esh3}
\end{enumerate}
Then for every $L$-vertex $u$ the flow entering $[u]$ is uniformly distributed over $W_{y,[u]} \subseteq [u]$ where 
  $|W_{y, [u]}|$ is independent of $y$.
\end{lemma}

\begin{proof}
We first use hypotheses (\ref{esh2}),(\ref{esh3}) of \lemref{e:symmetry} to show that for every flow $p_y$ and  for 
every $L$-vertex $u$, the incoming flow $p_y(u^{-})$ to $u$ is either $0$ or $\alpha_y([u]) >0$,
that is it depends only on the equivalence class of $u$. We then use transitivity and 
hypothesis~(\ref{esh3}) of \lemref{e:symmetry} to reach the conclusion 
of the lemma.

Let $V_t$ be the set of vertices at level $t$ and fix a flow $p_y$.  
The proof is then by induction on the level $t$ on a stronger statement for every $\sigma,\sigma'\in S_n$ and 
$L$-vertices $u\in V_{t}$ and $v,v'\in V_{t+1}$:
\begin{equation}\label{ind1}
[p_y((u,v))> 0 \text{ and } p_y((\sigma(u),v'))> 0] \implies p_y((u,v))=p_y((\sigma(u),v')),
\end{equation}
\begin{equation}\label{ind2}
[p_y(\sigma(u)^{-})> 0 \text{ and } p_y(\sigma'(u)^{-})> 0] \implies p_y(\sigma(u)^{-})=p_y(\sigma'(u)^{-}).
\end{equation}

At level $t=0$, the statement is correct since the root is unique, has incoming flow $1$, and outgoing edges with flow 
$0$ or $1/g^+(\mathrm{root})$.

Assume the statements hold up to and including level $t$.  
Hypothesis~\ref{esh2} implies that when $p_y((u,v)) >0$ for $u \in V_t$, 
it satisfies $p_y((u,v))=p_y(u^{-})/g^+([u])$, and similarly  
$p_y((\tau(u),v'))=p_y(\tau(u)^{-})/g^+([u])$. Therefore, Equation~\ref{ind2} at level $t$ implies Equation~\ref{ind1} at level $t+1$.

We now turn to Equation~\ref{ind2} at level $t+1$.  Fix $v\in V_{t+1}$ and $\sigma,\sigma'\in S$ such that $\sigma(v)$ 
and $\sigma'(v)$ have positive incoming flows. Then 
\[
p_y(\sigma(v)^{-})=\sum_{u\in V_{t}} p_y((u,\sigma(v))) \quad\text{ and }\quad p_y(\sigma'(v)^{-})=\sum_{u\in V_{t}} 
p_y((u,\sigma'(v))).
\]
We will show that $p_y(\sigma(v)^{-})=p_y(\sigma'(v)^{-})$ by proving the following equality for every $u$
\[
\sum_{\tau\in S_n} p_y((\tau(u),\sigma(v))) = \sum_{\tau\in S_n} p_y((\tau(u),\sigma'(v))).
\]

By Equation~\ref{ind1} at level $t$, all nonzero terms in the respective sum are identical. By Hypothesis~\ref{esh3}, the 
number of nonzero terms is $g^{-}([u],[v])$ in both sums. Therefore the two sums are identical.

We now have concluded that the incoming flow to an $L$-vertex  $u$ 
is either $0$ or $\alpha_y([u])>0$.  This implies 
that the flow entering $u$ is uniformly distributed over some set $W_{y, [u]} \subseteq [u]$.  We now show that 
the size of this set is independent of $y$.

If the flow is transitive then $\alpha_y([u])$ is independent of $y$ and furthermore by the second statement of \lemref{lem:trans_cons}
applied to the level of $u$,
$\sum_{v \in [u]} p_y(v^-) = \sum_{v \in [u]} p_{y'}(v^-)$.  Thus the number of terms in each sum must be the same 
and $|W_{y, [u]}|$ is independent of $y$.
\end{proof}

\section{Algorithms}\label{sec:algo}
We first discuss some basic assumptions about the subgraph $H$.  Say that $H$ has $k$ vertices and minimum 
degree $d$.  First, we assume that $d \ge 1$, that is $H$ has no disconnected vertices.  Recall that we are testing if
$G$ contains $H$ 
as a subgraph, not as an induced subgraph.  Thus if $H'$ is $H$ with disconnected vertices removed, then $G$ will 
contain $H$ if and only if $G$ contains $H'$ and $n \ge k$.  Furthermore, the algorithms we give in this 
section behave monotonically with $k$, and so will have smaller complexity on the graph $H'$.  Additionally, we 
assume that $k \ge 3$ as if $k=2$ and $d=1$ then $H$ is simply an edge and in the case the complexity is known to be $\Theta(n)$ as it is equivalent to search on $\Theta(n^2)$ items. 

Thus let $H$ be a graph on vertex set $\{1,2,\ldots, k\}$, with $k\geq 3$ vertices.  We present two algorithms in this 
section for determining if a graph $G$ contains $H$.  
Following Belovs, we 
say that a stage \emph{loads} an edge $\{a,b\}$ 
if for all $L$-edges $(u,v)$ with flow in the stage, we have $\{a,b\} \in S(v) \setminus S(u)$.  
Both algorithms will use a subroutine, given in \secref{sec:algo1},
to load an induced subgraph of $H$.  For some integer $1\leq  u\leq k$, let $H_{[1,u]}$ be the subgraph of $H$ induced 
by vertices $1,2,\ldots,u$.  The first algorithm, given in \secref{sec:loadH}, will take $u = k$ and load $H$ directly; 
the second algorithm, given in \secref{sec:algo2}, will first load $H_{[1, k-1]}$, and then search for the missing vertex 
that completes $H$.

\subsection{Loading a subgraph of $H$}\label{sec:algo1}
Fix $ 1 \le u \le k$ and and let $e_1, \ldots, e_m$ be the edges of $H_{[1,u]}$, enumerated in some fixed order.
We assume that $m\geq 1$.  
For any positive input graph $G$, that is a graph $G$ which contains a copy of $H$, 
we fix $k$ vertices $a_1,a_2,\ldots,a_k$ such that $\{a_i,a_j\}$ is an edge of $G$ whenever $\{i,j\}$ is an 
edge of $H$.

We define a bit of terminology that will be useful.  For two sets $Y_1, Y_2 \subseteq [n]$,
we say that a bipartite graph between $Y_1$ and $Y_2$ is of type
$(\{(n_1,d_1), \ldots, (n_j,d_j)\} , \{(m_1,g_1), \ldots, (m_\ell,g_\ell)\})$ if
$Y_1$ has $n_i$ vertices of degree $d_i$ for $i = 1 , \ldots , j$, and 
$Y_2$ has $m_i$ vertices of degree $g_i$ for $i = 1 , \ldots , \ell$, and this is a complete 
listing of vertices in the graph, i.e. $|Y_1|=\sum_{i=1}^j n_i$ 
and $|Y_2|=\sum_{i=1}^\ell m_i$.

Vertices of our learning graph will be labeled by a $u$-partite graph $Q$ on disjoint sets 
$X_1, \ldots, X_{u} \subseteq [n]$.
The global structure of $Q$ 
will mimic the edge pattern of $H_{[1,u]}$.  
Namely, for each edge $e_t=\{i, j\}$ of $H_{[1,u]}$, there will be a bipartite graph $Q_{t}$ 
between $X_i$ and $X_j$ with a specified degree sequence.  
There are no edges between $X_i$ and $X_j$ if $\{i, j\}$ is not an edge of $H_{[1,u]}$.  
The mapping $S: V \rightarrow 2^{n \choose 2}$ from learning graph vertices to query indices 
returns the union of the edges of $Q_t$ for $t=1, \ldots, m$. 

We now describe the stages of our first learning graph.  Let $V_t$ denote the $L$-vertices at the beginning 
of stage $t$ (and so the end of stage $t-1$ for $t>0$).  The $L$-edges 
between $V_t$ and $V_{t+1}$ are defined in the obvious way---there is an $L$-edge between
$v_t \in V_t$ and $v_{t+1} \in V_{t+1}$ if the graph labeling 
$v_t$ is a subgraph of the graph labeling $v_{t+1}$.
We initially set the weight of all $L$-edges to be one, though some edges will be reweighted in the 
complexity analysis using \lemref{e:simple_costg}.
The root of the learning graph is labeled by the 
empty graph.

The algorithm depends on two parameters $r,s$ which will be optimized later.  The parameter 
$r \in [n]$ will control the number of vertices, and $s \in [0,1]$ the edge 
density, of graphs labeling the $L$-vertices.

\paragraph*{Learning graph $\G_1$:}

\paragraph*{Stage $0$: Setup (Figure~\ref{fig:stage0}).} $V_1$ consists of all $L$-vertices labeled by a 
$u$-partite graph $Q$ with color classes $A_1, \ldots, A_{u} \subseteq [n]$, each of size $r-1$.  
The edges will be the union of the edges in bipartite graphs $Q_1, \ldots, Q_m$, where
if $e_\ell=\{i, j\}$ is an edge of $H_{[1,u]}$, then $Q_\ell$ is a bipartite graph of type
$(\{(r-1-rs,rs) , (rs,rs-1)\} , \{(r-1-rs,rs) , (rs,rs-1)\} )$ between $A_i$ and $A_j$. 
The number of edges added in this stage is $O(sr^2)$.  Flow is uniform from the root of the learning graph, 
whose label is the empty graph,
to all $L$-vertices such that $a_1, \ldots, a_k \not \in A_i$ for $i = 1, \ldots, u$.

\paragraph*{Stage $t$ for $ t = 1, \ldots, u$: Load $a_t$ (Figures~\ref{fig:stage1} and~\ref{fig:stage1b}).} $V_{t+1}$ consists of all $L$-vertices labeled by a $u$-partite graph $Q$ with 
color classes $B_1, \ldots, B_t$, $A_{t+1}, \ldots A_{u}$, where $|B_i| = r$, and $|A_i| = r-1$.
The edges of $Q$ are the union of edges of bipartite graphs $Q_1, \ldots, Q_m$, where if 
$e_\ell=\{i, j\}$ then the type of $Q_\ell$ is given by the following cases:
\begin{itemize}
\item If $t <i<j$, then $Q_\ell$ is of type $(\{(r-1-rs,rs) , (rs,rs-1)\} , \{(r-1-rs,rs) , (rs,rs-1)\} )$ between $A_i$ and $A_j$.
\item If $i \leq t <j$, then $Q_\ell$ is of type $(\{(r-rs,rs) , (rs,rs-1)\} , \{(r-1,rs) \} )$ between $B_i$ and $A_j$.
\item If $i < j \leq t$, then $Q_\ell$ is of type $(\{(r,rs) \} , \{(r,rs) \} )$ between $B_i$ and $B_j$.
\end{itemize}
The number of edges added at  stage $t$ is $O(rs)$.
The flow is directed uniformly on those $L$-edges where the element added to $A_t$ is $a_t$ and 
none of the edges $\{a_i,a_j\}$ are present. 

\paragraph*{Stage $u+1$: Hiding (Figure~\ref{fig:stage2}).} Now we are ready to start loading edges $\{a_i, a_j\}$.  
If we simply loaded the edge $\{a_i, a_j\}$ 
now, however, it would be uniquely identified by the degrees of $a_i, a_j$ since only these vertices would have 
degree $rs+1$.  This means that for example at the last stage of the learning graph
the vertex ratio would be $\Omega(n^{k-1})$, no matter
what $r$ is.
Thus in this stage we first do a ``hiding'' step, adding edges so that half of the vertices in every set have degree $rs+1$.  

Formally, $V_{u+2}$ consists of all $L$-vertices labeled by a $u$-partite graph $Q$ with 
color classes $B_1, \ldots, B_{u}$, where $|B_i| = r$.
The edges of $Q$ are the union of edges of bipartite graphs $Q_1, \ldots, Q_m$,
where if $e_\ell=\{i, j\}$ then $Q_\ell$ is of type 
$(\{(r/2,rs) , (r/2,rs+1)\} , \{(r/2,rs) , (r/2,rs+1) \} )$ between $B_i$ and $B_j$.
The number of edges added in this stage is $O(r)$.
The flow is directed uniformly to those $L$-vertices where for every $e_\ell=\{i, j\}$, 
both $a_i$ and $a_j$ have degree $rs$ in $Q_\ell$. 

\paragraph*{Stage $u+ t+1$ for $ t = 1, \ldots, m$: Load $\{a_i,a_j\}$ if $e_t= \{i,j\}$ (Figure~\ref{fig:stage2b}).} 
Take an $L$-vertex at the beginning of stage $u+t+1$ whose 
edges are the union of bipartite graphs $Q_1, \ldots, Q_m$.  In stage $u+t+1$ only $Q_t$ 
will be modified, by adding single edge $\{b_i,b_j\}$ where 
$b_i\in B_i$ and $b_j\in B_j$ have degree $rs$ in $Q_t$.  The flow is directed uniformly along those 
$L$-edges where $b_i=a_i$ and $b_j=a_j$.  

Thus at the end of stage $u+m+1$, the $L$-vertices are labeled by the edges in the union of bipartite graphs $Q_1, \ldots, Q_m$ each of type
$(\{(r/2-1,rs) , (r/2+1,rs+1)\} , \{(r/2-1,rs) , (r/2+1,rs+1) \} )$.  
The incoming flow is uniform over those $L$-vertices where $a_i \in B_i$ for $i= 1, \ldots, u$, and
if  $e_\ell=\{i, j\}$ then the edge $\{a_i, a_j\}$ is 
present in $Q_\ell$ for $\ell = 1, \ldots, m,$ and both $a_i, a_j$ have degree $rs+1$ in $Q_\ell$.

\begin{figure}[H]
\centerline{\includegraphics[width=300pt]{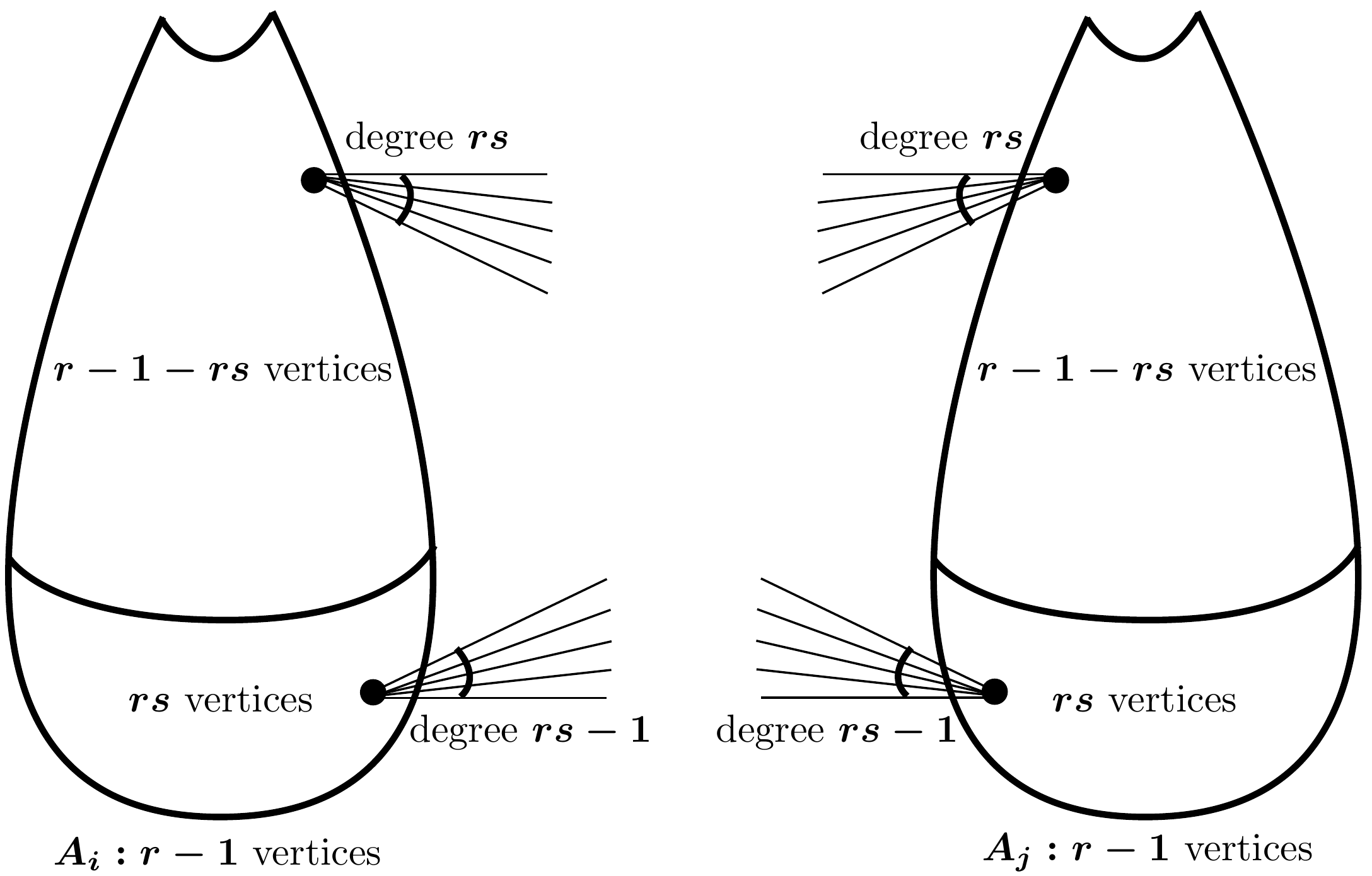}}
\caption{\textbf{Stage~$0$:} Edges added to $G_\ell$ when $e_\ell=\{i, j\}$ is an edge of $K$.
The flow is uniform to instances with  $a_1, \ldots, a_k \not \in A_i$ for $i = 1, \ldots, k-1$.
\label{fig:stage0}}
\end{figure}
\begin{figure}[H]
\centerline{\includegraphics[width=350pt]{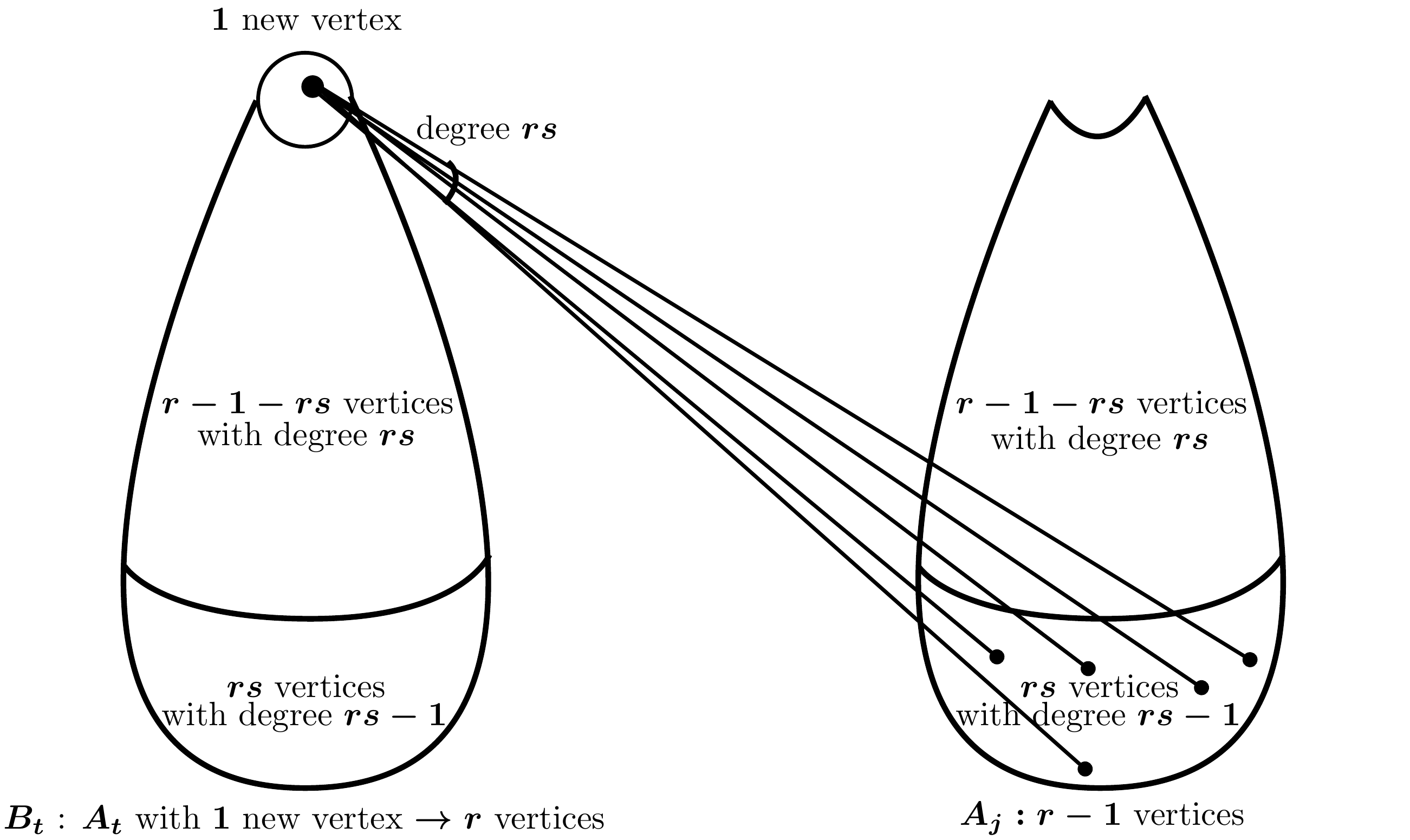}}
\caption{\textbf{Stage~$t$ for $ t = 1, \ldots, u$:} 
$rs$ added edges in some $G_\ell$ at stage~$t$,
when $e_\ell=\{t, j\}$ with $t<j$.
See Figure~\ref{fig:stage1b} for the case 
$e_\ell=\{i, t\}$ with $i<t$.
(No edge is added to $G_\ell$ at stage $t$ when $e_\ell=\{i, j\}$ with $t\neq i$ and $t\neq j$.)
The added edges are between the new vertex of $A_t$
and the $rs$ vertices in $A_j$, respectively $B_i$, of degree $(rs-1)$. 
The flow is directed to instances where the new vertex of $A_t$ is $a_t$.
\label{fig:stage1}}
\end{figure}
\begin{figure}[H]
\centerline{\includegraphics[width=350pt]{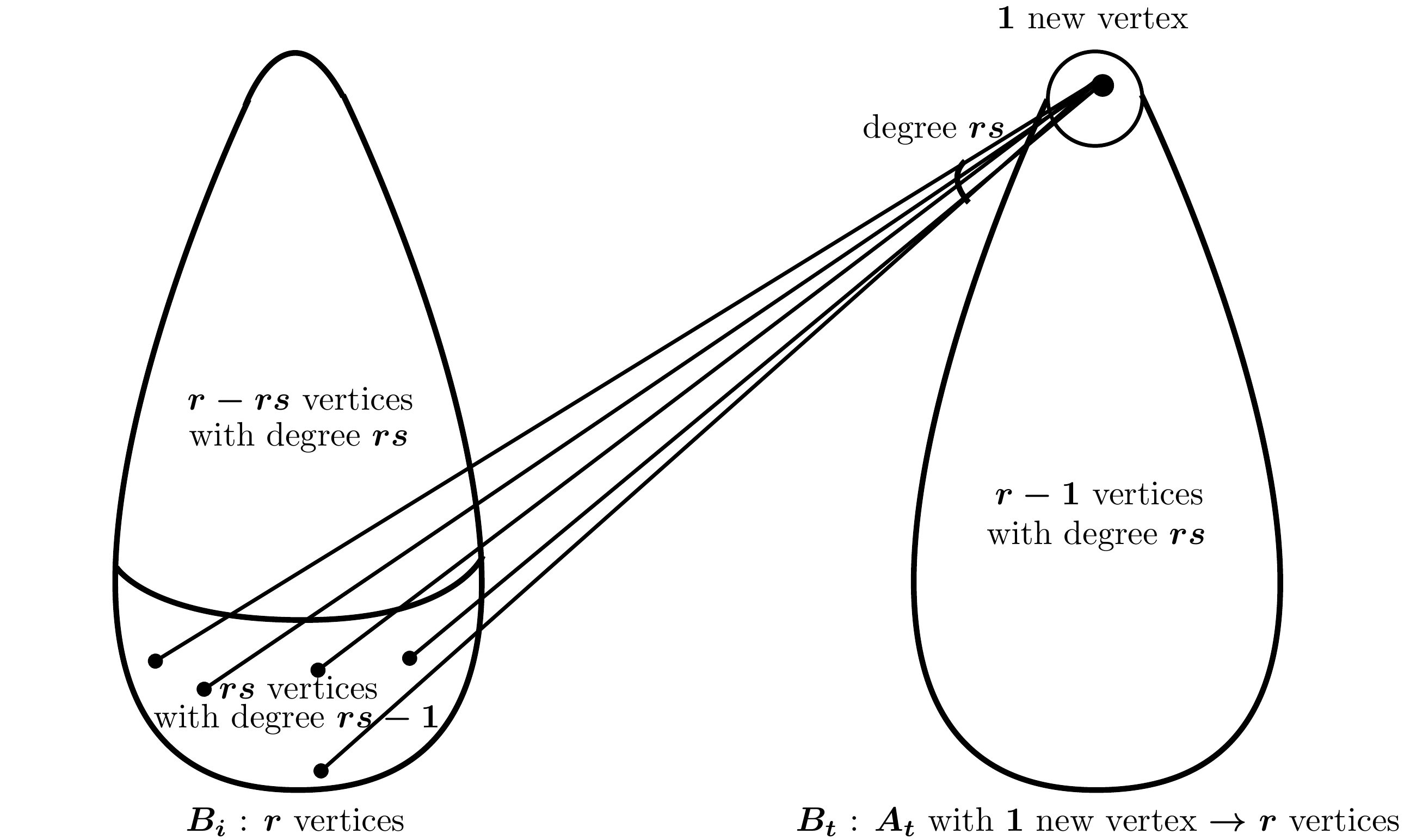}}
\caption{\textbf{Stage~$t$ for $ t = 1, \ldots, u$:} 
$rs$ added edges in some $G_\ell$ at stage~$t$,
when $e_\ell=\{i, t\}$ with $i<t$.
See Figure~\ref{fig:stage1} for the case $e_\ell=\{t, j\}$ with $t<j$.
(No edge is added to $G_\ell$ at stage $t$ when $e_\ell=\{i, j\}$ with $t\neq i$ and $t\neq j$.)
The added edges are between the new vertex of $A_t$
and the $rs$ vertices in $A_j$, respectively $B_i$, of degree $(rs-1)$. 
The flow is directed to instances where the new vertex of $A_t$ is $a_t$.
\label{fig:stage1b}}
\end{figure}
\begin{figure}[H]
\centerline{\includegraphics[width=300pt]{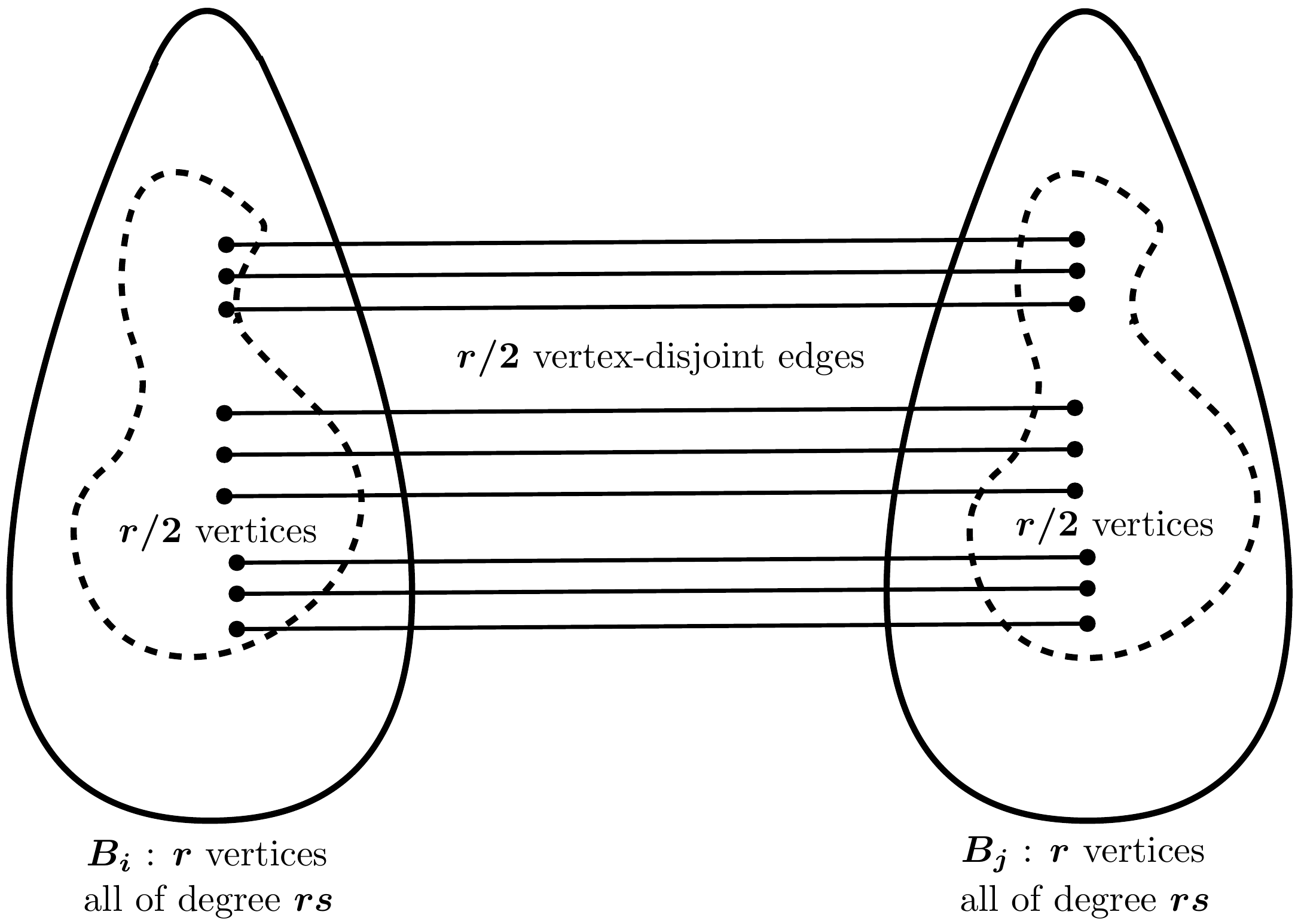}}
\caption{\textbf{Stage~$u+1$:} 
We add $r/2$  vertex-disjoint edges to $G_\ell$ when $e_\ell=\{i, j\}$ is an edge of $K$.
The flow is directed to instances where the degrees of $a_i$ and $a_j$ remain $rs$ in $G_\ell$.
\label{fig:stage2}}
\end{figure}
\begin{figure}[H]
\centerline{\includegraphics[width=300pt]{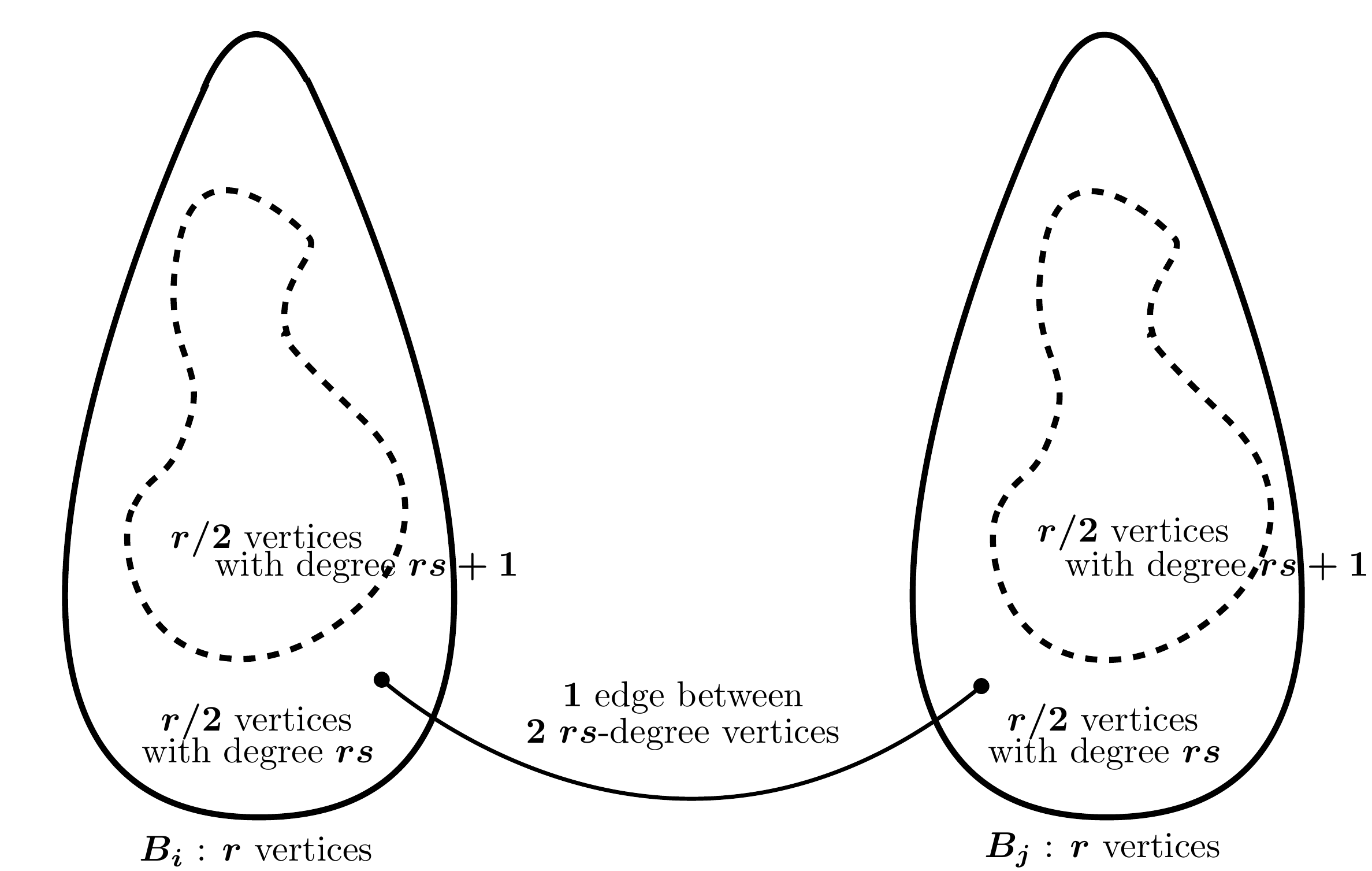}}
\caption{\textbf{Stage~$u+1+t$ for $ t = 1, \ldots, m$:}
Let $e_t=\{i,j\}$. Then a single edge is added to $Q_t$ between
two vertices $b_i\in B_i$ and $b_j\in B_j$ of degree $rs$ in $Q_t$.
The flow is directed to instances where $b_i=a_i$ and $b_j=a_j$.
\label{fig:stage2b}}
\end{figure}

\paragraph*{Complexity analysis of the stages}
Note that for an input graph $y$ containing a copy of $H_{[1,u]}$ the definition of 
flow depends only on the vertices $a_1, \ldots, a_u$ that span $H$.  As for any two graphs $y,y'$ containing $H$ 
there is a permutation $\tau$ mapping a copy of $H$ in $y$ to a copy of $H$ in $y'$ we see that $S_n$ acts 
transitively on flows.  

Furthermore, by construction of our learning graph, from a vertex $v$ with $p_y(v^-) > 0$, flow is directed uniformly 
to $g$ out of $d$ many neighbors, where $g,d$ depend only on the stage, not $y$ or $v$.  Additionally, by symmetry 
of the flow, hypothesis~(2) of \lemref{e:symmetry} is also satisfied.  We will invoke \lemref{e:simple_costg} to evaluate the cost of each stage.  Hypothesis~(1) is satisfied by \lemref{lem:trans_cons}, hypothesis~(2) by \lemref{e:symmetry}, 
and hypothesis~(3) by construction of the learning graph.

\begin{itemize}

\item Stage 0: 
The set of $L$-vertices at the beginning of 
this stage is simply the root thus the vertex ratio (and maximum vertex ratio) is one.  The degree ratio 
can be upper bounded by 
$((n-k)/(n-kr-k))^{k}=O(1)$, as we will choose $r=o(n)$ and $k$ is constant.  
The length of this stage is $O(sr^2)$ and so its complexity is $O(sr^2)$.

\item Stage $ t$ for $ t = 1, \ldots u$: 
An $L$-vertex in $V_t$ will be used by the flow if and only if $a_i \in B_i$ for $i = 1 , \ldots, t-1$ and 
$a_i \not \in B_1, \ldots, B_t, A_{t+1}, \ldots, A_u$ for $i=t, \ldots, k$.  For any vertex $v \in V_t$ the probability over 
$\sigma \in S_n$ that $\sigma(v)$ satisfies the second event is constant thus the vertex ratio is dominated by 
the first event which has probability $O((r/n)^{t-1})$.  Thus the maximum vertex ratio is $O((n/r)^{t-1})$.
The degree ratio is $n$.  Since $O(sr)$ edges are added, the complexity 
is $O(sr \sqrt{n} (n/r)^{(t-1)/2})$.

\item Stage $u+1$: 
As above, an $L$-vertex in $V_{k+1}$ will be used by the flow if and only if $a_i\in B_i$ for $i=1,\ldots, u$.  
For any vertex $v \in V_{k+1}$ the probability over $\sigma$ that this is satisfied by $\sigma(v)$ is $O((r/n)^{u})$
therefore the maximum vertex ratio is $O((n/r)^{u})$.
For each $e_\ell=\{i, j\}$, half of the vertices in $B_i$ and half of the vertices in $B_j$ will 
have degree $rs$ in $Q_\ell$.
Therefore, the degree ratio 
is $4^m=O(1)$.
Since $O(r)$ edges are added, the complexity of this stage is therefore $O(r (n/r)^{u/2})$.

\item Stage $u+t+1$ for $ t = 1, \ldots m$: 
In every stage, the degree ratio  is $O(r^2)$.
An $L$-vertex is in the flow at the beginning of stage $u+t+1$ if the following two 
conditions are satisfied:
\begin{align}
a_i &\in B_i \text{ for } i = 1 , \ldots u,  \label{e:condition1} \\
\text{ if } e_\ell=\{i, j\} \text{ then } \{a_i, a_j\} &\in Q_\ell \text{ with } a_i, a_j \text{ of degree } rs+1 \text{ in } Q_\ell, 
\text{ for } \ell=1, \ldots, t-1. \label{e:condition2}
\end{align} 
The probability over $\sigma$ that $\sigma(v)$ satisfies \eqnref{e:condition1} is $\Omega((r/n)^{u})$.  
Among vertices in $[v]$ satisfying this condition, a further $\Omega(s^{t-1})$ fraction will satisfy 
\eqnref{e:condition2}.  This follows from \lemref{l:uniform} below, together with the independence 
of the bipartite graphs $Q_1, \ldots, Q_m$.  Thus the maximum vertex ratio is
$O( (n/r)^{u} s^{-(t-1)})$.
As only one edge is added at this stage, 
we obtain a cost of $O( r (n/r)^{u/2} s^{-(t-1)/2})$.

\end{itemize}

\begin{lemma}
\label{l:uniform}
Let $Y_1,Y_2$ be disjoint $r$-element subsets of $[n]$,
and let $(y_1,y_2)\in Y_1\times Y_2$.
Let $K$ be a bipartite graph between $Y_1$ and $Y_2$
of type $(\{(r/2-1,rs) , (r/2+1,rs+1)\} , \{(r/2-1,rs) , (r/2+1,rs+1) \} )$.
The probability over $\sigma \in S_n$ that the edge $\{y_1,y_2\}$ is in $\sigma(K)$
and both $y_1$ and $y_2$ are of degree $rs+1$, is at least $s/4$.
\end{lemma}
\begin{proof}
The degree condition is satisfied with probability at least $1/4$.  Given that the degree 
condition is satisfied, it is enough to show that for a bipartite graph $K'$ of type 
$(\{(r,rs)\} , \{(r,rs) \} )$ the probability over $\sigma \in S_n$ that $\sigma(K')$ contains the fixed edge $(y_1,y_2)$ 
is at least $s$, since $K$ is such a graph plus some additional edges. 

Because of symmetry, this probability doesn't depend on the choice of 
the edge,
let's denote it by $p$. Let $K_1, \ldots, K_c$ be an enumeration of all bipartite graphs isomorphic to $K'$. We 
will count in two different ways the cardinality $\chi$ of the set
$\{ (e,h) : e \in K_h \}$. Every $K_h$ contains $sr^2$ edges, therefore $\chi = csr^2$. On the other hand, every edge
appears in $pc$ graphs, therefore $\chi = r^2pc$, and thus $p = s$.
\end{proof}

\subsection{Loading $H$}\label{sec:loadH}
When $u=k$, the constructed learning graph determines if $H$ is a subgraph of 
the input graph, since a copy of $H$ is loaded on positive instances.
Choosing the parameters $s,r$ to optimize the total cost gives the following theorem.

\begin{theorem}
\label{t:main1}
Let $H$ be a graph on $k\geq 3$ vertices and $m \ge 1$ edges.
Then there is a quantum query algorithm for 
determining if $H$ is a subgraph of an $n$-vertex graph making 
$O(n^{2-2/(k+1)-k/((k+1)(m+1))})$ many queries.
\end{theorem}
\begin{proof}
By \thmref{t:lg_alg}, it suffices to show that the learning graph $\G_1$ has the 
claimed complexity.  We will use \lemref{l:stages} and upper bound the learning graph 
complexity by the sum of the costs of the stages.  As usual, we will ignore factors of $k$.  

The complexity of stage $0$ is:
\[
S' = O\left(sr^2\right) \enspace .
\]
The complexity of each stage $1, \ldots, k$, and also their sum, is 
dominated by the complexity of stage $k$: 
\[
U'=O\left(sr \sqrt{n} (n/r)^{(k-1)/2}\right) \enspace .
\]
The complexity of stage $k+1$ is:
\[
U''=O\left(r (n/r)^{k/2}\right) \enspace .
\]
Again, the complexity of each stage $k+2, \ldots, k+m+1$, and also their sum, 
is dominated by the complexity of stage $k+m+1$:
\[
U'''=O\left( r (n/r)^{k/2} s^{-(m-1)/2}\right) \enspace .
\]
Observe that $U''=O(U''')$.

Therefore the overall cost can be bounded by $S' + U' + U'''$.
Choosing $r=n^{1-1/(k+1)}$ makes $S'=U'$ for any value of $s$, as their  dependence on $s$
is the same.  When $s=1$ we have $U''' < S'=U'$ thus we can choose $s <1$ to balance all three 
terms. Letting $s=n^{-t}$ we have $S'=U'=O(n^{2-2/(k+1)-t})$ and
$U'''=O(n^{1+(k-2)/(2(k+1))+t(m-1)/2})$.  Making these equal gives 
$t=k/((k+1)(m+1))$, and gives overall cost $O(n^{2-2/(k+1)-t})$.  
\end{proof}

\subsection{Loading the full graph but one vertex}\label{sec:algo2}
Recall that $H$ is a graph on vertex set $\{1,2,\ldots, k\}$, with $k\geq 3$ vertices, $m \ge 1$ edges and minimum 
degree $d \ge 1$.  By renaming the vertices, if necessary, we assume that vertex $k$ has degree $d$.

Our second algorithm employs the learning graph
$\G_1$ of Section~\ref{sec:algo1} with $u=k-1$ to first load $H_{[1, k-1]}$.  This is then combined with search to 
find the missing vertex and a collision subroutine to verify it links with $H_{[1, k-1]}$ to form $H$.

Again, let $H_{[1,k-1]}$ be the subgraph of $H$ induced by vertices $1,2,\ldots,k-1$,
and let $e_1, \ldots, e_{m'}$ be the edges of $H_{[1,k-1]}$, enumerated in some fixed order.
Thus note that $m = m' + d$.
For any positive input graph $y$, 
we fix $k$ vertices $a_1,a_2,\ldots,a_k$ such that $\{a_i,a_j\}$ is an edge of $y$ whenever $\{i,j\}$ is an 
edge of $H$.  For notational convenience we assume that $a_k$ is of degree $d$ and connected to 
$a_1, \ldots, a_d$.

\paragraph*{Learning graph $\G_2$:}

\paragraph*{Stages $0,1,\ldots,k+m'$:} Learning graph $\G_1$ of Section~\ref{sec:algo1}.

\paragraph*{Stage $k+m'+1$:} We use search plus a $d$-wise collision subroutine to find a vertex 
$v$ and $d$ edges which link $v$ to $H_{[1, k-1]}$ to form $H$.  The learning graph for this subroutine is given in
\secref{sec:sub}.

\paragraph*{Complexity analysis of the stages}
All stages but the last one have been analyzed in \secref{sec:algo1}, therefore
only the last stage remains to study.
\begin{itemize}
\item Stage $k+m'+1$: Let $V_{k+m'+1}$ be the set of $L$-vertices at the beginning of stage $k+m'+1$.
We will evaluate the complexity of this stage in a similar fashion as we have done previously.  As $S_n$ 
acts transitively on the flows, by \lemref{lem:trans_cons} we can invoke \lemref{lem:conv_comb} and it suffices to 
consider the maximum of $C(E_{[u]}^\rightarrow)$ over equivalence classes $[u]$.  Furthermore, as we have 
argued in \secref{sec:algo1}, the learning graph also satisfies the conditions of \lemref{e:symmetry}, thus we can apply 
\lemref{e:simple_cost} to evaluate $C(E_{[u]}^\rightarrow)$.  The maximum vertex ratio over $[u]$ 
is $O(s^{-m'} (n/r)^{k-1})$.  As shown in \secref{sec:sub}, the complexity of the subroutine learning 
graph attached to each $v \in V_{k+m'+1}$ is at most $O(\sqrt{n} r^{d/(d+1)})$.  Thus by \lemref{e:simple_cost},
the complexity of this stage is  
\[
O\left(s^{-m'/2}\left(\frac{n}{r}\right)^{(k-1)/2} \sqrt{n} r^{d/(d+1)}\right) \enspace .
\]

\end{itemize}

Choosing the parameters $s,r$ to optimize the total cost gives the following theorem.
\begin{theorem}
\label{t:main2}
Let $H$ be a graph on $k\geq 3$ vertices with minimal degree $d\geq 1$ and $m$ edges.
Then there is a quantum query algorithm for 
determining if $H$ is a subgraph of an $n$-vertex graph making 
$O(n^{2-2/k-(2k-d-3)/(k(d+1)(m-d+2))})$ many queries.
\end{theorem}
\begin{proof}
By \thmref{t:lg_alg}, it suffices to show that the learning graph $\G_2$ has the 
claimed complexity.  We will use \lemref{l:stages} and upper bound the learning graph 
complexity by the sum of the costs of the stages.  As usual, we will ignore factors of $k$.  

The complexity of stage $0$ is:
\[
S' = O\left(sr^2\right) \enspace .
\]
The complexity of each stage $1, \ldots, k-1$, and also their sum, is 
dominated by the complexity of stage $k-1$: 
\[
U'=O\left(sr \sqrt{n} (n/r)^{(k-2)/2}\right) \enspace .
\]The complexity of stage $k$ is:
\[
U''=O\left(r (n/r)^{(k-1)/2}\right) \enspace .
\]
Again, the complexity of each stage $k+1, \ldots, k+m'$, and also their sum, 
is dominated by the complexity of stage $k+m'$:
\[
U'''=O\left( r (n/r)^{(k-1)/2} s^{-(m'-1)/2}\right) \enspace.
\]
Observe that $U''=O(U''')$.
Finally, denote the cost of stage $k+m'+1$ by
\[
C'=O\left(s^{-m'/2}\left(\frac{n}{r}\right)^{(k-1)/2} \sqrt{n} r^{d/(d+1)}\right) \enspace.
\]

Observe that $U'''=O(C')$, provided that $r^{1/(d+1)} s^{1/2} = O(n^{1/2})$.
The later is always satisfied since $ s\le1$, $r\le n$ and $d\ge 1$.
Therefore the overall cost can then be bounded by $S' + U' + C'$.
Choosing $r=n^{1-1/k}$ makes $S'=U'$ for any value of $s$, as their $s$ dependence 
is the same.  When $s=1$ we have $C' < S'=U'$ thus we can choose $s <1$ to balance all three 
terms. Letting $s=n^{-t}$ we have $S'=U'=O(n^{2-2/k-t})$ and
$C'=O(n^{2-2/k +1/(2k)-(k-1)/(k(d+1))+tm'/2})$.  Making these equal gives 
$t=(2k-d-3)/(k(d+1)(m'+2))$.  Since $k \ge 3$ we have $t > 0$ and thus $s < 1$.  
The overall cost of the algorithm is $O(n^{2-2/k-t})$.  Noting that $m=m'+d$ gives the statement of the theorem.
\end{proof}
Our main result is an immediate consequence of Theorem~\ref{t:main1} and Theorem~\ref{t:main2}.
\begin{theorem}
\label{t:main3}
Let $H$ be a graph on $k\geq 3$ vertices with minimal degree $d\geq 1$ and $m$ edges.
Then there is a quantum query algorithm for determining if $H$ is a subgraph of an $n$-vertex graph making 
$O(n^{2 - 2/k -t})$ many queries, where 
\[
t = \max\left\{\frac{k^2- 2(m+1)}{k(k+1)(m+1)}, \ \frac{2k - d - 3}{k(d+1)(m-d+2)} \right\}.
\]
\end{theorem}

\subsection{Graph collision subroutine} \label{sec:sub}
In this section we describe a learning graph for the graph collision subroutine that is used in the 
learning graph given in \secref{sec:algo2}.  For each vertex $v$ at the end of stage $k+m'$ we will 
attach a learning graph $\G_v$.  The root of $\G_v$ will be the label of $v$ and we will show that it has 
complexity $\sqrt{n}r^{d/(d+1)}$.  Furthermore for every flow $p_y$ on $\G_v$, the sinks of flow will be $L$-vertices 
that have loaded a copy of $H$.  We now describe $\G_v$ in further detail.

A vertex $v$ at the end of stage $k+m'$ is labeled by a 
$(k-1)$-partite graph $Q$ on color classes $B_1, \ldots, B_{k-1}$ of size $r$.  The edges of $Q$ are the union of the 
edges in bipartite graphs $Q_1, \ldots, Q_{m'}$ each of type $(\{(r/2-1,rs) , (r/2+1,rs+1)\} , \{(r/2-1,rs) , (r/2+1,rs+1) \} )$.
This will be the label of the root of $\G_v$.  

On $\G_v$ we define a flow $p'_y$ for every input $y$ such that $p_y(v^-) > 0$ in the learning graph loading 
$H_{[1,k-1]}$.  Say that $y$ contains a copy of $H$ and that vertices $a_1, \ldots, a_k$ span $H$ in $y$.  
For ease of notation, assume that  vertex $a_k$ (the degree $d$ vertex removed from $H$) is connected to 
$a_1, \ldots, a_d$.  Recall that the $L$-vertex $v$ will have flow if and only if
$a_i \in B_i, a_k \not \in B_i$ for $i=1, \ldots, k-1$ and if $e_\ell=\{i,j\}$ then the edge $\{a_i, a_j\}$ is present in 
$Q_\ell$ for $\ell=1, \ldots, m'$, and both $a_i, a_j$ have degree $rs+1$ in $Q_\ell$.  Thus for each such 
$y$ we will define a flow on $\G_v$.  The flow will only depend on $a_1, \ldots, a_k$.   The complexity of $\G_v$ will 
depend on a parameter $1\leq \lambda \leq r$, that we will optimize later.

\paragraph*{Stage 0:} Choose a vertex $u \not \in B_i$ for $i=1, \ldots k-1$ and load $\lambda$
edges between $u$ and vertices of degree $rs+1$ in $B_i$, for each $i=1,\ldots d$.  Flow is directed uniformly 
along those L-edges where $u=a_k$ and none of the edges loaded touch any of the $a_1, \ldots, a_{d}$.  

\paragraph*{Stage $t$ for $t=1, \ldots, d$:} Load an additional edge between $u$ and 
$B_t$.  The flow is directed uniformly along those L-edges where the edge loaded is 
$\{a_k, a_t\}$.  

\paragraph*{Complexity analysis of the stages}
\begin{itemize}
  \item Stage 0: We use \lemref{e:simple_cost}.  As the vertices at the beginning of this stage consist only of 
  the root, conditions (1) and (2) are trivially satisfied; Condition~(3) is 
  satisfied by construction.  Flow is present in all $L$-edges of this stage where $u=a_k$, which is a 
  $\Omega(1/n)$ fraction of the total number of $L$-edges.  Thus the degree ratio $d/g=O(n)$.  The length of the 
  stage is $\lambda$, giving a total cost of $\lambda \sqrt{n}$.
  \item Stage $t$ for $t=1, \ldots, d$: Let $V_t$ be the set of vertices at the 
  beginning of stage~$t$.  The definition of flow depends only on $a_1, \ldots, a_k$, thus $S_n$ acts transitively on 
  the flows.  Applying \lemref{lem:trans_cons} gives that $\{p'(y)\}$ is consistent with $[u]^+$ for $u \in V_t$.  Also by 
  construction the hypothesis of \lemref{e:symmetry} is satisfied, thus we are in position to use \lemref{e:simple_costg}.
  
  The length of each stage is $1$.  The out-degree of an $L$-vertex in stage $t$ is $O(r)$ while the flow uses just 
  one outgoing edge, thus the degree ratio $d/g=O(r)$.  Finally, we must estimate the fraction of vertices in $[u]$ 
  with flow for $u \in V_t$.  A vertex $u$ in $V_t$ has flow if and only if $a_k$ was loaded in stage $0$ and the 
  edges $\{a_k, a_i\}$ are loaded for $i=1, \ldots, t-1$.  The probability over $\sigma \in S_n$ that the first event holds
  in $\sigma(u)$ is $\Omega(1/n)$.  Given that $a_k$ has been loaded at vertex $u \in V_t$ the probability over 
  $\sigma$ that $\{a_k, a_i\} \in \sigma(u)$ is $\Omega(\lambda/r)$.  Thus we obtain that the maximum vertex ratio 
  at stage $t$ is $n (r/\lambda)^{t-1}$.  The complexity of stage $t$ is maximized when $t=d$, giving an overall 
  complexity $\sqrt{nr}(r/\lambda)^{(t-1)/2}$.
 \end{itemize}

The sum of the costs $\lambda \sqrt{n}$ and $\sqrt{nr}(r/\lambda)^{(t-1)/2}$ is minimized for $\lambda=r^{d/(d+1)}$
giving a cost of $O(\sqrt{n} r^{d/(d+1)})$.

\subsection{Comparison with the quantum walk approach}
It is insightful to compare the cost of the learning graph algorithm for finding a subgraph 
with the the algorithm of \cite{MSS07} using a quantum walk on the Johnson graph.  We saw 
in the analysis of the learning graph that there were three important terms in the cost, 
denoted $S',U',C'$.  In the quantum 
walk formalism there are also three types of costs: setup, 
aggregated update, and aggregated checking, which 
we will denote by $S,U,C$.  
When the walk is done on the 
Johnson graph with vertices labeled by $r$-element subsets these costs are
\begin{align*}
S&=r^2 \\
U&= \left(\frac{n}{r}\right)^{(k-1)/2} r^{3/2} \\
C&=\left(\frac{n}{r}\right)^{(k-1)/2} \sqrt{n} r^{d/(d+1)}.
\end{align*}
Here $d$ is the minimal degree of a vertex in $H$. 

Here there is only one parameter, and in general $r$ cannot be chosen to make all three terms equal.  In the case of triangle 
finding ($k=3, d=2$), the choice $r=n^{3/5}$ is made.  
This makes $S=n^{1.2}$ and $U=C=n^{1.3}$.  In the general case of finding $H$, the choice 
$r=n^{1-1/k}$ is made, giving the first and second terms equal to $n^{2-2/k}$ and the third term 
$C=n^{2-1/k(1+k/(d+1)+(d-1)/2(d+1))}$.  Thus $C < S=U$ even for the largest possible value 
$d=k-1$.  Because of this, the analysis gives $n^{2-2/k}$ queries for any graph on $k$ vertices, 
independent of $d$.   


\section*{Acknowledgments}
We thank Aleksandrs Belovs for pointing out an error in the original version of the algorithm. We are very grateful to L\'aszl\'o Babai
whose insightful remarks made us realize the importance of distinguishing the two algorithms presented here. 
We also thank the anonymous referees for their helpful comments for improving the presentation of the paper.

\bibliographystyle{alpha}   
\bibliography{clique}

\end{document}